\documentclass[12pt,english,reqno]{amsart}
\usepackage{algorithm,algorithmic}
\usepackage{amsmath,amssymb,amsthm,mathtools}
\usepackage{enumerate}
\usepackage{tikz}
\usetikzlibrary{arrows,automata,backgrounds,calendar,mindmap,trees}
\usepackage{xcolor}
\usepackage{graphicx}
\usepackage[colorinlistoftodos]{todonotes}
\usepackage[margin = 1in]{geometry}
\usepackage[english]{babel}

\newcommand{\M}{\boldsymbol{M}}

\def\Hb{\mathbf{H}}
\def\E{\mathcal{E}}
\def\w{\mathbf{w}}
\def\wb{w_{\infty}}

\def\R{\mathbb{R}}

\def\Z{\mathcal{Z}}
\def\D{\mathcal{D}}

\def\Rb{\mathbf{R}}

\renewcommand{\pmb}{\boldsymbol}
\newcommand{\mbf}{\mathbf}
\newcommand{\krfw}{k^r_{\textrm{fw}}}
\newcommand{\krbw}{k^r_{\textrm{bw}}}

\newtheorem{theorem}{Theorem}[section]
\newtheorem{proposition}[theorem]{Proposition}
\newtheorem{corollary}[theorem]{Corollary}
\newtheorem{lemma}[theorem]{Lemma}
\newtheorem{remark}[theorem]{Remark}

\newtheorem{definition}[theorem]{Definition}

\usepackage{subfig}
\usepackage{titletoc}
\usepackage{hyperref}
\hypersetup{
    colorlinks,
    citecolor=cyan,
    filecolor=blue,
    linkcolor=blue,
    urlcolor=black
}
\title[Generalized potential games]{Generalized potential games}
\author{M. H. Duong}
\address{School of Mathematics, University of Birmingham, B15 2TT, UK.}
\email{H.Duong@bham.ac.uk}
\author{T. H. Dang-Ha}\thanks{Dang contributed to this work when he was at the University of Oslo}
\address{Google, Z\"{u}rich, Switzerland}
\email{danghien@google.com}
\author{Q. B. Tang}
\address{Institute for Mathematics and Scientific Computing, University of Graz, Austria}
\email{quoc.tang@uni-graz.at, baotangquoc@gmail.com}

\author{H. M. Tran}
\address{Data Analytics Department, Esmart Systems, 1783 Halden, Norway.}
\email{hoang.minh.tran@esmartsystems.com, hoangtm.fami@gmail.com}
\begin{document}
\maketitle
\begin{abstract}
In this paper, we introduce a notion of generalized potential games that is inspired by a newly developed theory on generalized gradient flows. More precisely, a game is called generalized potential if the simultaneous gradient of the loss functions is a nonlinear function of the gradient of a potential function. Applications include a class of games arising from chemical reaction networks with detailed balance condition. For this class of games, we prove an explicit exponential convergence to equilibrium for evolution of a single reversible reaction. Moreover, numerical investigations are performed to calculate the equilibrium state of some reversible chemical reactions which give rise to generalized potential games.
\end{abstract}
\tableofcontents
\section{Introduction}
In the seminal work \cite{monderer1996potential} Monderer and Shapley introduced a fundamental concept of (multi-player) potential games\footnote{In this paper, potential games covers both exact and weighted potential games as defined in \cite{monderer1996potential}, cf. Section \ref{sec: potential games}.}. They are a class of games where the incentive of all players to change their strategy can be expressed via a single common function, the potential function. From the mathematical point of view, a game is potential if the simultaneous gradient of the loss functions are equal to the gradient of the potential function (with respect to the corresponding controllable variable) scaled by weight factors. Using the potential function, the existence of pure-strategy Nash equilibria and the convergence to these equilibria have been shown \cite{monderer1996potential,Rosenthal1973}. Furthermore, the computation of a Nash equilibrium is reduced to the solution of an optimization problem: finding a pure-strategy Nash equilibrium is equivalent to finding a local minimum of the potential function.  In practice, Nash equilibria are often computed using the standard gradient descent method. Because of these distinctive features, potential games have been studied intensively in theoretical research as well as have been found in a cornucopia of practical applications in economics \cite{monderer1996potential,Rosenthal1973, AV2001, OS2005, MW2008} and other disciplines such as artificial intelligence and computer vision \cite{Miller1991, Yu1995}, theoretical
computer science, computational social science and sociology \cite{Montanari2009, Montanari2010} and wireless networks \cite{Yamamoto2015}. More recently, methods and concepts from game theory have been increasingly used in machine learning. For example, the celebrated Generative Adversarial Networks can be formulated as a two-player zero-sum game between two neural networks (the discriminator and the generator) \cite{Goodfellow-et-al2014} and learning with bandit feedback can be cast in the framework of potential games \cite{HeliouCohenMetikoulos2017}. We refer the reader to recent articles \cite{Schuurmans2016,Letcher2019,Zhou2019} and references therein for further information .


\medskip
The ultimate goal of this paper is to introduce a generalization of potential games inspired by recent developments in the theory of generalized gradient flows in the field of partial differential equations. In a classical gradient system, the evolution decreases a driving (entropy) functional ``as fast as possible'' given by a (quadratic) dissipation potential expressing a linear relation between rates and driving forces \cite{Onsager1931a, Onsager1931b, Morrison1986, Ottinger2005}. Solutions to a gradient flow system can also be iteratively approximated using the steepest gradient descent method \cite{AGS2008}. Recently, a concept of generalized gradient flows has been introduced where the dissipation potential is not necessarily quadratic \cite{Mielke2011, Mielke2014}. This generalization significantly enlarges the class of classical gradient systems and finds new applications in different areas, such as in the modeling of materials (for example, for plasticity and ferromagnetism) and in chemical reaction networks. Furthermore, generalized gradient flows arise in two natural ways: via large-deviation principles from the hydrodynamic limit of microscopic many-particle systems \cite{DuongPhD2014,Mielke2014} or as suitable multiscale limits (evolutionary $\Gamma$-convergence) of classical gradient systems \cite{Mielke2016a,LMPR2017}. The existence of the potential function as well as the linear relation between the simultaneous gradient of the loss functions and the gradient of the potential function in potential games resemble the existence of the driving functional and the linear relation between rates and driving forces. Motivated by this newly developed theory and the similarities between potential games and classical gradient flows, we introduce a notion of \textit{generalized potential games} where the simultaneous gradient of loss functions is a nonlinear function of the gradient of a potential function. This notion enlarges significantly the class of games which can be studied using tools from (generalized) gradient flows. Moreover, by providing an alternative criterion of potential games via the symmetrizability of the Hessian matrix, we give a necessary and sufficient condition for a game to be generalized potential. 

To show the applicability of our concept, we consider a class of loss functions arising from chemical reaction network theory under detailed balance condition. This theory has been employed in many places, from classical areas such as chemistry \cite{feinbergbook} to evolutionary game theory \cite{Veloz2014, Velegol2018}, social sciences~ \cite{DittrichFenizio2007,DittrichWinter2008} and recently in machine learning~\cite{KayalaBaldi2011,SimmReiher2018,Fooshee-etal2018}. Here we bring it to classical game theory.


We show in this paper that if a chemical reaction system satisfies the detailed balance condition, then the game arising from this system is generalized potential. This enables us to use the projected gradient steepest descent method to calculate the chemical equilibrium. This is closely related to the trend to equilibrium of the evolution for the chemical reactions. We thus prove an explicit convergence to equilibrium for the case of a single reversible reaction with an arbitrary number of chemical species by the Bakry-Emery strategy. Numerical simulations show the similarity of the evolution and the convergence to stationary point of the potential function in the projected steepest gradient descent method. Therefore, we expect that a discrete version of our proof based on Bakry-Emery strategy will shed light on the convergence of numerical methods for generalized potential games.

\medskip
To sum up, the main contributions of our paper are the following:
\begin{enumerate}
\item we introduce a notion of generalized potential games enlarging significantly the class of classical potential games;
\item we provide a connection between potential games and symmetrizable matrices, and use this connection to characterize generalized potential games;
\item we calculate numerically the equilibrium sate of several reversible reaction systems, and compare them with the evolution of the system.
\end{enumerate}

\textbf{The rest of the paper is organized as follows}. In the next section, we recall the definition of potential games and provide a connection to symmetrizable matrices. The notion of generalized potential games is introduced in Section \ref{subsec:concept} where we also provide a necessary and sufficient condition for a game to be generalized potential. Using this condition, we show in Section \ref{sec: chemical reaction games} that a game arising from a chemical reaction system is generalized potential if the system satisfies the detailed balance condition. An explicit exponential convergence to equilibrium for the evolution of a single reversible reaction is also given. In Section \ref{sec:numerics}, several games from specific chemical reactions are studied numerically. Finally, Section \ref{sec: summary} is devoted to conclusion and outlook.
\section{Potential games and symmetrizable matrices}
\label{sec: potential games}
In this section, we recall the concept of potential games and relevant results, which have attracted a lot of attention recently in the machine learning community \cite{balduzzi2018mechanics,Letcher2019}. Moreover, we establish a connection between potential games and symmetrizable matrices.
\begin{definition}\label{def_game}
A $n$-player differentiable game is a set of players $[n]=\{1,\ldots,n\}$ and twice continuously differentiable loss functions\footnote{also called payoff/utility functions in the literature. We use the terminology loss function since we consider the minimization problem.}
\begin{align*}
\ell_i: D_1\times \ldots\times D_n &\longrightarrow\R
\\ \w=(\w_1,\ldots,\w_n)&\mapsto \ell_i(\w),
\end{align*}
where $\w_i\in D_i\subset \R^{d_i}$ with $\sum_{i=1}^n d_i=d$. The objective of each player is to minimize its loss function, that is to seek
$$
\w_i^\ast=\mathrm{argmin}_{\w_i\in D_i} \ell_i(\w).
$$
The simultaneous gradient is the combination of gradients of the losses with respect to the controllable parameters of the corresponding players:
\begin{equation}
\label{eq: simultaneous gradient}
\xi(\w)=(\nabla_{\w_1}\ell_1,\ldots,\nabla_{\w_n}\ell_n).
\end{equation}
A point $\w^\ast\in \R^d$ is called a fixed point (Nash stationary point) of the game if
$$
\xi(\w^\ast)=(\nabla_{\w_1}\ell_1(\w^\ast),\ldots,\nabla_{\w_n}\ell_n (\w^\ast))=0.
$$
\end{definition}
An important class of differentiable games are potential games introduced in the seminal paper \cite{monderer1996potential}.
\begin{definition}
\label{def: PG}
A game is called a potential game if and only if there exist a function $\phi:\R^d\to\R$ and positive numbers $\{\alpha_i\}_{i=1}^n$ such that
\begin{equation}
\label{eq: potential game}
\alpha_i\nabla_{\w_i}\ell_i=\nabla_{\w_i}\phi\qquad\forall i=1,\ldots, n.
\end{equation}
\end{definition}
Note that in \cite{monderer1996potential} games satisfying condition \eqref{eq: potential game} are called weighted potential games and in particular exact potential games when $\alpha_i=1$ for all $i$. An advantage of potential games is that we only need to work with one single potential function $\phi$ instead of $n$ different loss functions $\{\ell_i\}_{i=1}^n$. The following theorem provides a necessary and sufficient condition to determine a potential game and a useful formula to compute the potential function.
\begin{theorem}\cite[Theorem 4.5]{monderer1996potential}
\label{thm: MondererShapley1996}
A differentiable game is potential if and only if there exist positive numbers (weights) $\{\alpha_i\}_{i=1}^n$ such that
\begin{equation}
\label{eq: potential condition}
\alpha_i\nabla^2_{\w_i\w_j}\ell_i=\alpha_j\nabla^2_{\w_i\w_j}\ell_j\qquad\forall\, i,j.
\end{equation}
Under this condition, there exists a potential $\phi$ that satisfies
\begin{equation}
\label{eq: potential phi}
\phi(z_1)-\phi(z_0)=\sum_{i=1}^n \alpha_i\int_0^1 z_i'(t)\cdot\nabla_{\w_i}\ell_i(z(t)))\,dt \qquad\forall\, z_0,z_1\in\Z,
\end{equation}
where $z:[0,1]\rightarrow \Z$ is a continuously differentiable path in $\Z$ that connects two strategy profiles $z_0$ and $z_1$, that is $z(0)=z_0$ and $z(1)=z_1$.

Furthermore, the potential function is uniquely defined up to an additive constant.
\end{theorem}
 Using the potential function, the existence and attainability of a pure-strategy Nash equilibrium have been proved \cite{monderer1996potential, Rosenthal1973}. Furthermore, it can be numerically computed using the steepest gradient descent method under a broad range of conditions \cite{Lee2016}. There is a vast literature about potential games, and we refer the reader to monographs \cite{Nisan2007,La2016} for great expositions on this topic. Identifying a potential game by showing the existence of the weights in \eqref{eq: potential condition} is a non-trivial task. It would be more useful to seek a condition that can be verified intrinsically using only the given loss functions. In the next section, we discuss this issue.

\medskip
Next, we establish a connection between potential games and symmetrizable matrices. In order to determine whether a game is potential using Definition \ref{def: PG}, one needs to show the existence of the weighted $(\alpha_i)_{i=1,\ldots, n}$ which is usually difficult. One technique to overcome this issue is to explore the connection between the Hessian of the game (see \eqref{eq: Hessian}) and the so-called symmetrizable matrices, which provides an algorithm to check if a game is potential.

\begin{definition}
\label{def: symmetrizable}
An $n$-by-$n$ matrix $A$ is said to be symmetrizable if there exists an invertible diagonal matrix $D$ such that $DA$ is symmetric. In this case the matrix $D$ is called a symmetrizer of $A$.
\end{definition}
An example of symmetrizable matrix is $A=\begin{pmatrix}
2& 1\\
\frac{1}{2}&-2
\end{pmatrix}$. This matrix is obviously not symmetric but is symmetrizable because  
$$
\begin{pmatrix}
1&0\\
0&2
\end{pmatrix}A=
\begin{pmatrix}
1&0\\
0&2
\end{pmatrix}\begin{pmatrix}
2& 1\\
\frac{1}{2}&-2
\end{pmatrix}=\begin{pmatrix}
2&1\\
1&-4
\end{pmatrix}.
$$
For more information on symmetrizable matrices we refer the reader to \cite{BAKSALARY1981, VENKAIAH1988}. The construction/existence of a symmetrizer $D$ in Definition \ref{def: symmetrizable} is often nontrivial. It is more useful to seek a necessary and sufficient condition to verify the symmetrizability of a matrix without invoking externally the existence of the weights. The following lemma provides such a condition.
\begin{lemma}\cite{Hearon1953}
\label{lem: symmetrizable matrix}
A matrix $A=(a_{ij})$ is symmetrizable if and only if the following conditions are satisfied
\begin{enumerate}
\item $a_{ij}=0$ implies $a_{ji}=0$ for all $1\leq i\leq j\leq n$.
\item $a_{i_1i_2}a_{i_2i_3}\ldots a_{i_ki_1}=a_{i_2i_1}a_{i_3i_2}\ldots a_{i_1i_k}$ for all $k\geq 3$ and $i_1,i_2,\ldots,i_k\in\{1,2,\ldots,n\}$.
\end{enumerate}
\end{lemma}
Thanks to Lemma \ref{lem: symmetrizable matrix}, an algorithm can be provided to determine whether a matrix is symmetrizable or not. Moreover, one can also compute a corresponding symmetrizer, see e.g. \cite[Algorithm 1]{Dias2016}.

\medskip
We now seek the connection between potential games and symmetrizable matrices. Consider an $n$-player differentiable game in Definition \ref{def_game}. We define the Hessian of the game to be
\begin{equation}
\label{eq: Hessian}
\Hb(\w):=\begin{pmatrix}
\nabla^2_{\w_1\w_1}\ell_1&\nabla^2_{\w_1\w_2}\ell_1&\ldots&\nabla^2_{\w_1\w_n}\ell_1\\
\nabla^2_{\w_2\w_1}\ell_2&\nabla^2_{\w_2\w_2}\ell_2&\ldots&\nabla^2_{\w_2\w_n}\ell_2\\
\vdots&\vdots&\ddots&\vdots\\
\nabla^2_{\w_n\w_1}\ell_n&\nabla^2_{\w_n\w_2}\ell_n&\ldots&\nabla^2_{\w_n\w_n}\ell_n\\
\end{pmatrix}.
\end{equation}
Condition \eqref{eq: potential condition} is then equivalent to the requirement that the matrix $M\Hb(\w)$ is symmetric where $M = \mathrm{diag}(\alpha_1,\ldots, \alpha_n)$, which means that $\Hb(\w)$ is symmetrizable. In particular, a game is exact potential iff the Hessian matrix is symmetric. 

From this observation and Lemma \ref{lem: symmetrizable matrix} we obtain the following necessary and sufficient conditions, that are based only on the loss functions, for a game to be potential. 
\begin{theorem}
\label{thm: WPG}
A differentiable game is a potential game if and only if the Hessian matrix $\Hb(\w)$ is symmetrizable for all $\w$.  That is, if and only if the following conditions hold for all $\w$:
\begin{enumerate}
\item $\Hb_{ij}(\w)=0$ implies $\Hb_{ji}(\w)=0$ for all $1\leq i\leq j\leq n$.
\item $\Hb_{i_1i_2}(\w)\Hb_{i_2i_3}(\w)\ldots \Hb_{i_ki_1}(\w)=\Hb_{i_2i_1}(\w)\Hb_{i_3i_2}(\w)\ldots \Hb_{i_1i_k}(\w)$ for all $k\geq 3$ and $i_1,i_2,\ldots,i_k\in\{1,2,\ldots,n\}$.
\end{enumerate}
\end{theorem}
We consider an important class of games recently studied in \cite{Minarchenko2018} where the loss functions are given by
\begin{equation}
\label{eq: nonquadratic games}
\ell_i(\w)=\hat{\ell}_{i}(\w_i)+\sum_{1\leq j\neq i\leq n}\w_i^T C\w_j,
\end{equation}
where the $\hat{\ell}_i$ are generally non-quadratic functions. Note that $C$ is simply the Hessian matrix of $\ell_i$ with respect to $\w=(\w_1,\ldots, \w_n)$. According to \eqref{eq: potential condition} (see also \cite[Proposition 3.1]{Minarchenko2018}) we obtain the following lemma.
\begin{lemma} 
\label{lem: weighted potential games}
The game with loss functions $\ell_i$ given by \eqref{eq: nonquadratic games} is potential if and only if there exist positive numbers $\{\alpha_i\}_{i=1}^n$ such that
\begin{equation}
\alpha_i C_{ij}=\alpha_j C_{ji}^T.
\end{equation}
\end{lemma}
As a direct application of Theorem \ref{thm: WPG} we obtain the following necessary and sufficient conditions for the game with loss functions \eqref{eq: nonquadratic games} to be potential. Note that this condition is directly based on the matrix $C$ only. 
\begin{corollary}The game 
with loss functions $\ell_i$ given by \eqref{eq: nonquadratic games} is weighted potential if and only if $C$ is symmetrizable. 
\end{corollary}
Thanks to Lemma \ref{lem: symmetrizable matrix}, \cite[Algorithm 1]{Dias2016} provides an algorithm to determine whether the game with loss functions given by \eqref{eq: nonquadratic games} is a potential game or not.
\section{Generalized potential games}
\label{sec: gPG}
In this section, we introduce a generalization of potential games which we call \textit{generalized potential games} where the simultaneous gradient of the loss functions is a \textit{nonlinear} function of the gradient of a potential function. In the first part, we give in details the concept, as well as its characterization. The second part provides an example of generalized potential games arising from chemical reaction network theory.

\subsection{Generalized potential games}
\label{subsec:concept}

From \eqref{eq: simultaneous gradient}, we notice that \eqref{eq: potential game} is equivalent to
\begin{equation}
\label{eq: M}
\xi=M^{-1}\nabla_{\w}\phi\quad\text{where}\quad M=\mathrm{diag}(\alpha_1,\ldots,\alpha_n).
\end{equation}
Thus to define a potential game we need three elements: an underlying Hilbert space $\Z$ (in this case, $\Z=\R^d$), a driving functional $\phi:\Z\to\R$, and an operator that maps the tangential bundle $T\Z = Z\times Z$ to $\Z$ (in this case it is a matrix multiplication with $M^{-1}$). 
Motivated by this observation and the theory of generalized gradient flow developed recently \cite{Mielke2011, Mielke2014}, we provide a concept of generalized potential game. 
\begin{definition}[\cite{Mielke2014}]
\label{def: dissipation}
A function $\Psi: \Z\times T\Z\rightarrow \R$ is called a dissipation potential if for all $z\in\Z$ 
\begin{enumerate}[(i)]
\item $\Psi(z,\cdot)$ is convex in the second argument,
\item $\min \Psi(z,\cdot)=0$ and
\item $\Psi(z,0)=0$.
\end{enumerate}
\end{definition}
Given a dissipation potential $\Psi$, we recall that the Legendre-Fenchel dual (convex dual),  $\Psi^{\ast}:\Z\times T^{\ast}\Z\rightarrow \R$, of $\Psi$, is defined as follows: 
\begin{equation}
\Psi^{\ast}(z,\zeta)=\sup_{s\in T_z\Z}\big\{\langle \zeta,s\rangle-\Psi(z,s) \big\}\quad\text{for}\quad \zeta\in T_z^{\ast}\Z.
\end{equation}
The following lemma provides useful properties of Legendre-Fenchel duality.
\begin{lemma}\cite{fenchel1949}
\label{lem: LF} Let $\Z$ be a reflexive Banach space and $\Psi: \Z\rightarrow \R\cup \{\infty\}$ be proper, convex and lower semicontinuous. Then for every $\xi\in\Z^\ast$ and every $s\in \Z$ the following statements are equivalent
\begin{equation}
\label{eq: LF equivalent}
(i)~ \xi \in\D\Psi(s)\quad \Longleftrightarrow\quad (ii)~ s\in \D\Psi^\ast(\xi)\quad \Longleftrightarrow \quad (iii)~ \Psi(s)+\Psi^\ast(\xi)=\langle \xi,s\rangle.
\end{equation}
\end{lemma}
Note that in this section, we will use the notation $\D$ for the Fr\'echet
differential to indicate that the setup might be applicable to a more general space than $\R^d$.
\begin{definition}
\label{def: GPG}
A differentiable game with loss functions $\{\ell_i:\Z\rightarrow\R\}_{i=1}^n$ is called a generalized potential game if there exist a dissipation potential $\Psi:T\Z \rightarrow\R$ and a potential function $\E:\Z \rightarrow\R$ such that
\begin{equation}
\label{eq: generalizedPG}
\xi(\w)=\D_\zeta\Psi^\ast(\D_{\w}\E)
\end{equation}
where $\xi(\w)$ is the simultaneous gradient defined in \eqref{eq: simultaneous gradient}.
\end{definition}
According to Lemma \ref{lem: LF} , \eqref{eq: generalizedPG} is equivalent to
\begin{equation}
\label{eq: equiv generalizedPG}
\D_\w\E(\w)=\D_s\Psi(\xi(\w)).
\end{equation}
Using the duality between $\Psi$ and $\Psi^\ast$, we have the following equivalences for Conditions $(ii)$ and $(iii)$ in Definition \ref{def: dissipation}: 
\begin{align*}
\min \Psi(\cdot)=0\qquad&\Longleftrightarrow\qquad\Psi^\ast(0)=0
\\ \Psi(0)=0\qquad&\Longleftrightarrow \qquad \min\Psi^\ast=0.
\end{align*}
Thus both $\Psi\geq 0$ and $\Psi^\ast \geq 0$ attain $0$ at $0$. Condition \eqref{eq: generalizedPG} implies that
\[
\xi(\w)=0\qquad\Longleftrightarrow\qquad\D_\zeta\Psi^\ast(\D_{\w}\E)=0\qquad\Longleftrightarrow \quad\D_{\w}\E=0,
\]
so that, as in a potential game, in a generalized potential game the negative of the simultaneous gradient $\xi$ is the same as that of the potential function $\E$. A potential game is a special case of generalized potential games when the dissipation potential is quadratic (i.e., its derivative $\D_\zeta\Psi^\ast(\zeta)$ is linear)
\begin{equation}
\label{eq: Psi-potential game}
\Psi(s)=\frac{1}{2}s^T M s,~\Psi^\ast(\zeta)=\frac{1}{2}\zeta^T M^{-1}\zeta, \quad \D_\zeta \Psi^\ast(\zeta)=M^{-1} \zeta, \quad \E=\phi,
\end{equation}
where $M=\mathrm{diag}(\alpha_1,\ldots,\alpha_n)$.
Thus potential games correspond to the case where $\D_\zeta\Psi^\ast(\zeta)$ is a linear function of $\zeta$. A generalized potential game is then a natural extension of a potential game since it allows \textit{nonlinear games}, i.e., $\D_\zeta \Psi^\ast(\zeta)$ can be a nonlinear function of $\zeta$.
The following theorem extends Theorem \ref{thm: MondererShapley1996} to generalized potential games:
\begin{theorem}\
A game is generalized potential iff the matrix $\D^2\Psi(\xi(\w))\Hb(\w)$ is symmetric. In this case, a potential $\E$ satisfies
\begin{equation}
\label{eq: potential}
\E(z_1)-\E(z_0)=\int_0^1 z'(t)\cdot\D_{s}\Psi(\xi(z(t)))\,dt \qquad\forall\, z_0,z_1\in\Z,
\end{equation}
where $z:[0,1]\rightarrow \Z$ is a continuously differentiable path in $\Z$ that connects two strategy profiles $z_0$ and $z_1$. That is $z(0)=z_0$ and $z(1)=z_1$.
\end{theorem}
\begin{proof}
Condition \eqref{eq: equiv generalizedPG} means that
\[
\D_{\w_i}\E=\D_{s_i}\Psi(\xi(\w)) \quad \forall i=1,\ldots,n.
\]
Applying $\D_{\w_j}$ to both sides, noting that $\xi(\w)=(\nabla_{\w_1}\ell_1,\ldots,\nabla_{\w_n}\ell_n)$, we get
\[
 \D_{\w_j}\big(\D_{s_i}\Psi(\xi(\w))\big)=\D^2_{\w_j\w_i}\E=\D^2_{\w_i\w_j}\E=\D_{\w_i}\big(\D_{s_j}\Psi(\xi(\w))\big)\quad\forall i,j,
\]
i.e.,
\begin{equation}
\label{eq: equaility}
\sum_{k=1}^n\D^2_{s_ks_i}\Psi(\xi(\w))\nabla^2_{\w_j\w_k}\ell_k=\sum_{k=1}^n\D^2_{s_ks_j}\Psi(\xi(\w))\nabla^2_{\w_i\w_k}\ell_k.
\end{equation}
Using the Hessian matrix $\Hb(\w)$ defined in \eqref{eq: Hessian}, equality \eqref{eq: equaility} can be written as
\begin{equation}
\label{eq: symmetry}
[\D^2\Psi(\xi(\w))\Hb(\w)]_{ij}=[\D^2\Psi(\xi(\w))\Hb(\w)]_{ji}\quad\forall i,j.
\end{equation}
i.e., the matrix $\D^2\Psi(\xi(\w))\Hb(\w)$ is symmetric.

Next we prove \eqref{eq: potential}. Let $z_0,z_1\in \Z$ and let $z:[0,1]\rightarrow\Z $ be a continuously differntiable path connecting $z_0$ and $z_1$. We have
\begin{align*}
\E(z_1)-\E(z_0)&=\E(z(1))-\E(z(0))
\\&\overset{(*)}{=}\int_0^1 \frac{d}{dt}\E(z(t))\,dt
\\&\overset{(**)}{=}\int_0^1 z'(t)\cdot\D_{\w}\E(z(t))\,dt
\\&\overset{(***)}{=}\int_0^1 z'(t)\cdot\D_{s}\Psi(\xi(z(t)))\,dt
\end{align*}
where we have applied the fundamental theorem of calculus, the chain rule and condition \eqref{eq: equiv generalizedPG} to obtain $(*), (**)$ and $(***)$ respectively.
\end{proof}
\begin{remark}
In a potential game, using \eqref{eq: Psi-potential game}, we have 
\[
\D^2\Psi(\xi(\w))=M=\mathrm{diag}(\alpha_1,\ldots,\alpha_n).
\]
Thus condition \eqref{eq: symmetry} becomes
\[
\alpha_i\nabla^2_{\w_i\w_j}\ell_i=\alpha_j\nabla^2_{\w_i\w_j}\ell_j\quad \forall i,j.
\]
So, we recover the result of Theorem \ref{thm: MondererShapley1996}
\end{remark}
\subsection{Examples in chemical reaction network theory}
\label{sec: chemical reaction games}\hfil

In this section, we will show that the generalized potential games corresponding to a class of loss functions arising from general reversible chemical reaction networks. 
\subsubsection{Chemical reaction networks}
Consider $K$ chemical species $X_1,\ldots, X_K$ subject to $R$ reversible reactions:
\begin{equation}\label{Reactions}
\alpha_1^r X_1+\ldots+\alpha^r_K X_K\xrightleftharpoons[k^r_{\textrm{bw}}]{\,k^r_{\textrm{fw}}\,} \beta_1^r X_1+\ldots+\beta^r_K X_K
\end{equation}
for $r=1,\ldots, R$. The coefficients $\alpha_k^r, \beta_k^r\in \mathbb{N}_0$  are stoichiometric coefficients and $k^r_{\textrm{fw}},k^r_{\textrm{fw}}>0$ are reaction rate constants. Let $\mathbf{c}=(c_1,\ldots, c_K)\in\mathbb{R}^K_+$ denote the concentration vector of species $X_1,\ldots,X_K$. Then the reaction rate equation that characterizes the time-evolution of the concentrations is given by
\begin{equation}
\label{eq: CRE}
\dot{\mathbf{c}}=-\Rb(\mathbf{c})\quad\text{with}~~ \Rb(\mathbf{c}):=\sum_{r=1}^R \Big(k^r_{\textrm{fw}}\mathbf{c}^{\pmb{\alpha}^r}-k^r_{\textrm{bw}}\mathbf{c}^{\pmb{\beta}^r}\Big)\Big(\pmb{\alpha}^r-\pmb{\beta}^r\Big)
\end{equation}
where $\pmb{\alpha}^r = (\alpha^r_i)_{i=1,\ldots, K}$ and $\pmb{\beta}^r = (\beta^r_i)_{i=1,\ldots, K}$, and we use the convention
\begin{equation}\label{convention}
	\mathbf{c}^{\pmb{\alpha}^r} = \prod_{k=1}^Kc_i^{\alpha^r_i}.
\end{equation}
A state $\mbf{w}_\infty \in \mathbb{R}_{+}^n$ is called a chemical equilibrium for \eqref{eq: CRE} if $\mbf{R}(\mbf{w}_\infty) = 0$. Moreover, $\mbf{w}_\infty$ is called a detailed balanced equilibrium if
\begin{equation}\label{equilibrium}
	\krfw \mbf{w}_\infty^{\alpha^r} = \krbw\mbf{w}_{\infty}^{\beta^r} \quad \text{ for all } \quad r = 1,\ldots, R.
\end{equation}
Note that the solution to \eqref{eq: CRE} is positive, i.e. it lies in the positive orthant $\mathbb R_+^n$, as long as the initial data is positive. Moreover, it satisfies
\begin{equation*}
	\mbf{c}(t) = \mbf{c}(0) + \sum_{r=1}^{R}(\pmb{\beta}^r - \pmb{\alpha}^r)\int_0^t\left(\krfw \mbf{c}^{\pmb{\alpha}^r}(s) - \krbw \mbf{c}^{\pmb{\beta}^r}(s)\right)ds.
\end{equation*}
Therefore, if we define the Wegscheider matrix by $W = (\pmb{\beta}^r - \pmb{\alpha}^r)_{r=1,\ldots, R} \in \mathbb R^{n\times R}$ then it follows that
\begin{equation*}
	\mbf{c}(t) \in \mbf{c}(0) + \mathrm{range}(W) \quad \text{ for all } \quad t\geq 0.
\end{equation*}
The affine space $(\mbf{c}(0) + \mathrm{range}(W))\cap \mathbb R_+^n$ is therefore called the {\it positive compatibility class} (or shortly {\it compatibility class}) corresponding to $\mbf{c}(0)$. The following lemma gives the existence and uniqueness of a positive detailed balanced equilibrium.
\begin{proposition}\cite{feinbergbook}\label{pro:equilibrium}
	If the chemical reaction network \eqref{Reactions} has a detailed balanced equilibrium, then all other equilibria are also detailed balanced. In this case we say that the network \eqref{Reactions} is detailed balanced.
	
	If the chemical reaction network \eqref{Reactions} is detailed balanced, then there exists a unique positive equilibrium in each compatibility class.
\end{proposition}
\begin{remark}
	It is remarked that there might exist (possibly infinitely) many {\it boundary equilibria}, i.e. a state $\mbf{c}^*\in \partial\mathbb{R}_+^n$ which satisfies \eqref{equilibrium} within each compatibility class.
\end{remark}

%

\subsubsection{Generalized potential games from chemical reaction networks}
Consider an $n$-player game where the loss functions are given by

\begin{subequations}
\begin{align}
\ell_1(\w)&=\sum_{r=1}^R \Big(\frac{k^r_{\textrm{fw}}}{\alpha_1^r+1}w_1^{\alpha_1^r+1}w_2^{\alpha_2^r}\cdots w_n^{\alpha^r_n}-\frac{k^r_{\textrm{bw}}}{\beta_1^r+1}w_1^{\beta_1^r+1}w_2^{\alpha_2^r}\cdots w_n^{\beta_n^r}
\Big)\Big(\alpha_1^r-\beta_1^r\Big)\label{eq: loss1}
\\&\vdots\nonumber
\\\ell_n(\w)&=\sum_{r=1}^R \Big(\frac{k^r_{\textrm{fw}}}{\alpha_n^r+1}w_1^{\alpha_1^r}w_2^{\alpha_2^r}\cdots w_n^{\alpha^r_n+1}-\frac{k^r_{\textrm{bw}}}{\beta_n^r+1}w_1^{\beta_1^r}w_2^{\alpha_2^r}\cdots w_n^{\beta_n^r+1}
\Big)\Big(\alpha_n^r-\beta_n^r\Big). \label{eq: lossn}
\end{align}
\end{subequations}
where $\w=(w_1,\ldots,w_n)$, 
 $\pmb{\alpha}^r:=(\alpha_1^r,\ldots,\alpha_n^r)$ and $\pmb{\beta}^r:=(\beta_1^r,\ldots,\beta_n^r)$.

\medskip
The following theorem provides an important class of generalized potential games. This theorem is motivated from a recent theory on the generalized gradient flows of chemical reactions \cite{Mielke2011, Mielke2014}.
\begin{theorem}
\label{thm: Example generalised game}
Suppose there exists ${\w}_\infty\in\mathbb{R}^n_+$ such that \eqref{equilibrium} is satisfied.
Then the game with loss functions \eqref{eq: loss1}-\eqref{eq: lossn} is a generalised potential game. 
\end{theorem}
\begin{proof}
We show that the game can be viewed as a generalized potential game by showing the existence of a potential function $\E$ and a dissipation function $\Psi^*$. These functions have been derived in the context of generalized gradient flow structure of chemical reactions \cite{Mielke2011}. Let us define the dissipation $\Psi^*$ and the potential function $\E$ by
\begin{align}
 \E(\w)&:=\sum_{i=1}^n w_i\Big(\log\big(\frac{w_i}{w_{i\infty}}\big)-1\Big),\\
 \Psi^*(\w, \pmb{\mu})&=\sum_{r=1}^R\frac{\kappa_r}{2}\ell\Big(\frac{\w^{\pmb{\alpha}^r}}{\w_{\infty}^{\pmb{\alpha}^r}},\frac{\w^{\pmb{\beta}^r}}{\w_{\infty}^{\pmb{\beta}^r}}\Big)\Big(\pmb{\mu}\cdot(\pmb{\alpha}^r-\pmb{\beta}^r)\Big)^2,
\end{align}
where 
\begin{equation}
\label{eq: ell}
\ell(a,b)=\begin{cases}
\frac{a-b}{\log a-\log b}~~\textrm{for}~~a\neq b,\\
b~~\textrm{for}~~ a=b.
\end{cases}
\end{equation}
Note that $\Psi^*$ can be written as the quadratic form
\begin{equation}
\Psi^*(\w,\pmb{\mu})=\frac{1}{2}\pmb{\mu}^T \Hb(\w)\pmb{\mu}~~\text{where}~~\Hb(\w)=\sum_{r=1}^R\kappa_r\ell\Big(\frac{\w^{\pmb{\alpha}^r}}{\w_{\infty}^{\pmb{\alpha}^r}},\frac{\w^{\pmb{\beta}^r}}{\w_{\infty}^{\pmb{\beta}^r}}\Big) (\pmb{\alpha}^r-\pmb{\beta}^r)(\pmb{\alpha}^r-\pmb{\beta}^r)^T.
\end{equation}
Hence $\nabla_{\w}\Psi^*(\w,\pmb{\mu})=\Hb(\w)\pmb{\mu}$. Therefore, 
$$
\nabla_{\w}\Psi^*(\w,\nabla_{\w}\E(\w))=\Hb(\w)\nabla_{\w}\E(\w).
$$
We compute
\begin{align*}
&\nabla_{\w}\E(\w)=\begin{pmatrix}
\log\Big(\frac{w_1}{w_{1\infty}}\Big)\\
\vdots\\
\log\Big(\frac{w_n}{w_{n\infty}}\Big)
\end{pmatrix}=:\log\Big(\frac{\w}{\w_{\infty}}\Big)
\\& \mathcal D_{\pmb{\mu}}\Psi^*(\w,\pmb{\mu})=\sum_{r=1}^R\kappa_r\ell\Big(\frac{\w^{\pmb{\alpha}^r}}{\w_{\infty}^{\pmb{\alpha}^r}},\frac{\w^{\pmb{\beta}^r}}{\w_{\infty}^{\pmb{\beta}^r}}\Big)\Big(\pmb{\mu}\cdot(\pmb{\alpha}^r-\pmb{\beta}^r)\Big)\Big(\pmb{\alpha}^r-\pmb{\beta}^r\Big).
\end{align*}
Since
$$
\log\Big(\frac{\w_n}{\w_{n\infty}}\Big)\cdot\Big(\pmb{\alpha}^r-\pmb{\beta}^r\Big)
=\log\Big(\frac{\w}{\w_{\infty}}\Big)^{\pmb{\alpha}^r-\pmb{\beta}^r},\quad \ell\Big(\frac{\w^{\pmb{\alpha}^r}}{\w_{\infty}^{\pmb{\alpha}^r}},\frac{\w^{\pmb{\beta}^r}}{\w_{\infty}^{\pmb{\beta}^r}}\Big)=\frac{\w^{\pmb{\alpha}^r}/\w_{\infty}^{\pmb{\alpha}^r}-\w^{\pmb{\beta}^r}/\w_{\infty}^{\pmb{\beta}^r}}{\log\Big(\frac{\w}{\w_{\infty}}\Big)^{\pmb{\alpha}^r-\pmb{\beta}^r} }
$$
we obtain
\begin{align}
\label{eq: DPsi}
 \mathcal D_{\pmb{\mu}}\Psi^*(\w,\nabla_{\w}\E(\w))&=\sum_{r=1}^R\kappa_r\Big(\frac{\w^{\pmb{\alpha}^r}}{\w_\infty^{\pmb{\alpha}^r}}-\frac{\w^{\pmb{\beta}^r}}{\w_\infty^{\pmb{\beta}^r}}\Big)\Big(\pmb{\alpha}^r-\pmb{\beta}^r\Big)\nonumber
\\&=\sum_{r=1}^R\Big(k^r_{\textrm{fw}}\w^{\pmb{\alpha}^r}-k^r_{\textrm{bw}}\w^{\pmb{\beta}^r}\Big)\Big(\pmb{\alpha}^r-\pmb{\beta}^r\Big),
\end{align}
where we have used the detailed balance condition \eqref{equilibrium} to obtain the last equality. On the other hand, we have for $j=1,\ldots, n$,
\begin{align*}
\nabla_{w_j}\ell_j(\w)&=\sum_{r=1}^R \Big(k^r_{\textrm{fw}}\w^{\pmb{\alpha}^r}-k^r_{\textrm{bw}}\w^{\pmb{\beta}^r}\Big)\Big(\alpha_j^r-\beta_j^r\Big).
\end{align*}
Using the simultaneous gradient in \eqref{eq: simultaneous gradient} we can write these in a compact form as
\begin{equation}
\label{eq: xi}
\xi (\w)=\sum_{r=1}^R \Big(k^r_{\textrm{fw}}\w^{\pmb{\alpha}^r}-k^r_{\textrm{bw}}\w^{\pmb{\beta}^r}\Big)\Big(\pmb{\alpha}^r-\pmb{\beta}^r\Big)    
\end{equation}
From \eqref{eq: DPsi} and \eqref{eq: xi} we get
\begin{equation}
\label{eq: generalized relation}
\xi(\w)=\mathcal{D}_{\pmb{\mu}}\Psi^*(\w,\nabla_{\w}\E(\w))=\Hb(\w)\nabla_{\w}\E(\w).
\end{equation}
Hence the game is a generalized potential game.
\end{proof}
\begin{remark}
	At this point, we remark that the condition in Lemma \ref{lem: symmetrizable matrix} is also a characterization of the detailed balance condition. Indeed, consider a first order reaction network of $n$ chemical species $S_1, \ldots, S_n$ which react through the following first order reactions
	\begin{equation*}
		S_i \underset{a_{ji}}{\overset{a_{ij}}{\leftrightharpoons}} S_j,
	\end{equation*}
	where $a_{ij}, a_{ji} \geq 0$ are reaction rate constants. Then the condition given in Lemma \ref{lem: symmetrizable matrix} is equivalent to the fact that this first order reaction network satisfies the detailed balance condition \cite{Hearon1953}.
\end{remark}
\begin{remark}
We notice that \eqref{eq: generalized relation} is a generalization of \eqref{eq: M} where the constant (and diagonal) matrix $M^{-1}$ is replaced by $\Hb(\w)$ which is dependent on $\w$. It should be mentioned that the pair $(\E, \Psi)$ is not uniquely determined. Another choice, which arises from large-deviation theory of interacting systems, is given by \cite{Mielke2014, MPR2016}:
\begin{align*}
\bar{\E}(\w)&:=\frac{1}{2}\sum_{i=1}^n w_i\Big(\log\big(\frac{w_i}{{w}_{i\infty}}\big)-1\Big),\\
\bar{\Psi}^\ast(\w,\pmb{\mu})&=\sum_{r=1}^R \Psi_r^\ast(\w,\pmb{\mu}),\quad \Psi_r^\ast(\w,\pmb{\mu})=2\kappa_r\w^{\frac{1}{2}\pmb{\alpha}^r+\frac{1}{2}\pmb{\beta}^r}\Big(\mathrm{cosh}\big(\pmb{\alpha}^r-\pmb{\beta}^r\big)\cdot\pmb{\mu}-1\Big).
\end{align*}
Note that $\bar{\E}=\frac{1}{2}\E$, and $\bar{\Psi}^\ast$ is a non-quadratic potential function.
\end{remark}

\medskip
The equilibrium of the loss functions defined in \eqref{eq: loss1}--\eqref{eq: lossn} agrees with the detailed balance equilibrium of the chemical reaction network presented in \eqref{Reactions}. Therefore, finding the equilibrium for the game defined in \eqref{eq: loss1}--\eqref{eq: lossn} relates strongly to the trend to equilibrium of the evolution system \eqref{eq: CRE}. The latter issue is well studied thanks to the recent advances concerning the large time behavior of chemical reaction networks.
The following exponential convergence to equilibrium for detailed balanced systems was proved in \cite{DFT17}\footnote{The results in \cite{DFT17} in fact covers a more general class called {\it complex balanced systems}.}.
\begin{proposition}\label{convergence}
	Assume that the chemical reaction network \eqref{Reactions} is detailed balanced and possesses no boundary equilibria. Then, for each positive initial data, the corresponding solution converges exponentially to the detailed balanced equilibrium within the same compatibility class.
\end{proposition}
\begin{remark}\label{remark:explicit}\hfil\
	\begin{itemize}
		\item[(i)] 
		The convergence to equilibrium for a system {\it with} boundary equilibria is related to the Global Attractor Conjecture \cite{Craciun}, which is still a famous open problem in chemical reaction network theory.
		
		\item[(ii)] The exponential convergence to equilibrium in Proposition \ref{convergence} was obtained via a contradiction argument, and therefore the convergence rate is not explicit. Up to our knowledge, the explicit convergence rate for \eqref{Reactions} is in general unknown. Nevertheless, for a single reversible reaction with an arbitrary number of chemical species, the convergence rate can be calculated explicitly. We demonstrate that result in Proposition \ref{pro:explicit} below.
		
		\item[(ii)] If we consider the spatial diffusion for each species in \eqref{Reactions}, then we obtain a reaction-diffusion system instead of a differential system. The proof of equilibration in this case is much more involved. We refer the interested reader to \cite{FT18}.
	\end{itemize}
\end{remark}

\begin{proposition}\label{pro:explicit}
	Let $\alpha_1, \ldots, \alpha_m, \beta_1, \ldots, \beta_n\in [1,\infty)$ and consider the reversible reaction
	\begin{equation}\label{reversible}
		\alpha_1 A_1 + \alpha_2 A_2 + \ldots + \alpha_mA_m \underset{k_f}{\overset{k_b}{\leftrightharpoons}} \beta_1 B_1 + \beta_2 B_2 + \ldots + \beta_nB_n.
	\end{equation}
	Denote by $a_i(t)$ and $b_j(t)$, $i\in \{1,\ldots, m\}$ and $j\in \{1,\ldots, n\}$ the concentrations of $A_i$ and $B_j$ at time $t>0$ respectively. Then for any positive initial data $(a_0,b_0)\in (0,\infty)^{m+n}$, the concentrations converge exponentially to equilibrium with {\normalfont computable rates}, i.e.
	\begin{equation*}
		\sum_{i=1}^m|a_i(t) - a_{i\infty}|^2 + \sum_{j=1}^n|b_j(t) - b_{j\infty}|^2 \leq C_0e^{-\lambda t}
	\end{equation*}
	where the constants $C_0$ and $\lambda$ can be {\normalfont explicitly computed.}
\end{proposition}
\begin{proof}
	The proof utilises the Bakry-Emery strategy, and we postpone it to Appendix \ref{appendix} in order to not interrupt the train of thought.
	
	\medskip
	We emphasize that although the convergence to equilibrium can be obtained by a contradiction argument (see e.g. \cite[Proposition 2.3]{DFT17}), the proof based on the Bakry-Emery method in Appendix \ref{appendix} has the advantage of providing an {\it explicit} rate of the convergence. Moreover, we believe that the convergence of a {\it discrete} version of this method leads to the convergence of algorithms finding equilibrium of the games defined in \eqref{eq: loss1}--\eqref{eq: lossn}. We leave this interesting point for future investigation.
\end{proof}

\section{Numerical investigations}\label{sec:numerics}
In this section we apply the projected steepest gradient descent method to numerically calculate the Nash equilibrium of generalized potential games arising from some specific chemical reactions. Moreover, we compare the behavior of the numerical schemes with the evolution of the chemical systems to see their close relation.

\medskip
Consider a general reversible chemical reaction 
\begin{equation}
\label{eq: general chemical reaction}
\alpha_1^r X_1+\ldots+\alpha^r_K X_K\xrightleftharpoons[k^r_{\textrm{bw}}]{\,k^r_{\textrm{fw}}\,} \beta_1^r X_1+\ldots+\beta^r_K X_K
\end{equation}
for $r=1,\ldots, R$. Suppose that the detailed balance condition \eqref{equilibrium} is satisfied. According to Proposition \ref{pro:equilibrium}, given an initial state, the reversible chemical reaction \eqref{eq: general chemical reaction} possesses a unique positive equilibrium state in the same compatibility class. To take into account the conservation of mass, we will apply the projected gradient descent method, see for instance \cite{Calamai1987,snyman2018practical}. Below we consider three concrete examples.

\medskip
\textbf{Example 1} ({\it A single reversible reaction}) Consider the following chemical reaction:
\begin{equation}
2NaCl+CaCO_{3}\rightleftharpoons Na_{2}CO_{3}+CaCl_{2}\label{eq:na2co3}
\end{equation}
Let $w_{1}=[NaCl],w_{2}=[CaCO_{3}],w_{3}=[Na_{2}CO_{3}],w_{4}=[CaCl_{2}]$. According to the law of mass action we obtain the following system of differential equations:
\begin{align*}
\dot{w}_{1} & =-2(k_{+}w_{1}^{2}w_{2}-k_{-}w_{3}w_{4}),\quad
\dot{w}_{2} =-(k_{+}w_{1}^{2}w_{2}-k_{-}w_{3}w_{4})\\
\dot{w}_{3} & =(k_{+}w_{1}^{2}w_{2}-k_{-}w_{3}w_{4}),\quad
\dot{w}_{4} =(k_{+}w_{1}^{2}w_{2}-k_{-}w_{3}w_{4}).
\end{align*}
The loss functions are given by
\begin{align*}
\ell_{1} & =2(\frac{1}{3} k_{+}w_{1}^{3}w_{2}-k_{-}w_{1}w_{3}w_{4}),\quad
\ell_{2} =(\frac{1}{2} k_{+} w_{1}^{2}w_{2}^{2}-k_{-}w_{2}w_{3}w_{4})\\
\ell_{3} & =-(k_{+}w_{1}^{2}w_{2}w_{3}-\frac{1}{2} k_{-}w_{3}^{2}w_{4}),\quad
\ell_{4} =-(k_{+}w_{1}^{2}w_{2}w_{4}-\frac{1}{2} k_{-}w_{3}w_{4}^{2}).
\end{align*}
The potential function is given by:
\begin{align*}
\mathcal{E}&=\sum_{i=1}^{4}w_i\Big(\log(w_i/w_{i\infty})-1\Big).
\end{align*}
Conservation of mass:
\begin{align*}
&w_{1}+2w_{3}  =C_{1}, && w_{1}+2w_{4}  =C_{2}\\
&w_{2}+w_{4}  =C_{3}, && w_{2}+w_{3}  =C_{4}
\end{align*}
Using the projected gradient method, the concentrations of all chemical
species in the corresponding chemical reaction converge to a close
neighborhood of the equilibrium point (Fig.\ref{chemical} (a,b)). In this case, the 
gradient vector is approximately 0, hence the KKT point is attained
(within certain numerical error), and all quantities remain constant
near equilibrium. 
\begin{figure}

\subfloat[$CaCO_{3}$ example]{\includegraphics[scale=0.5]{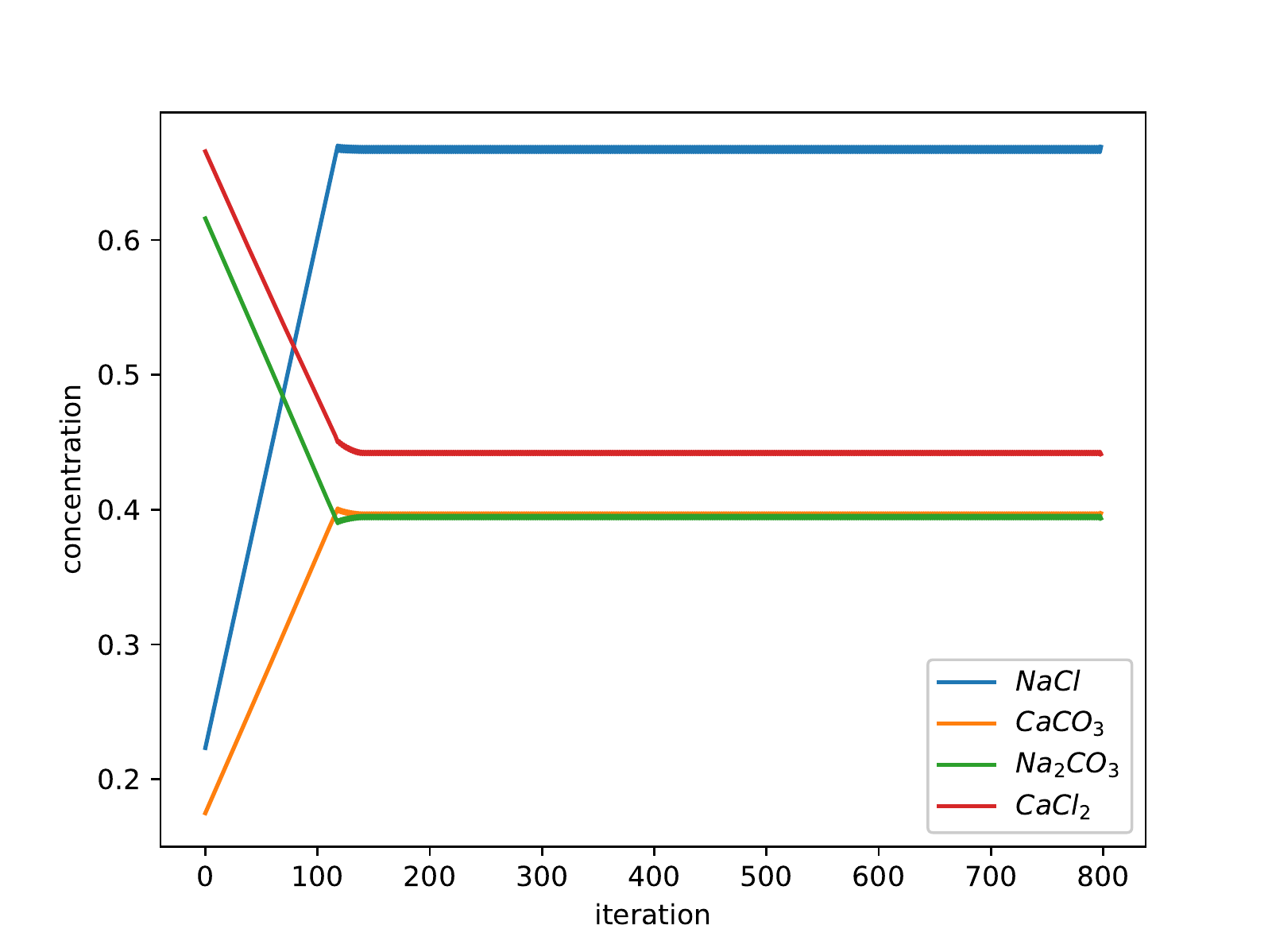}\includegraphics[scale=0.5]{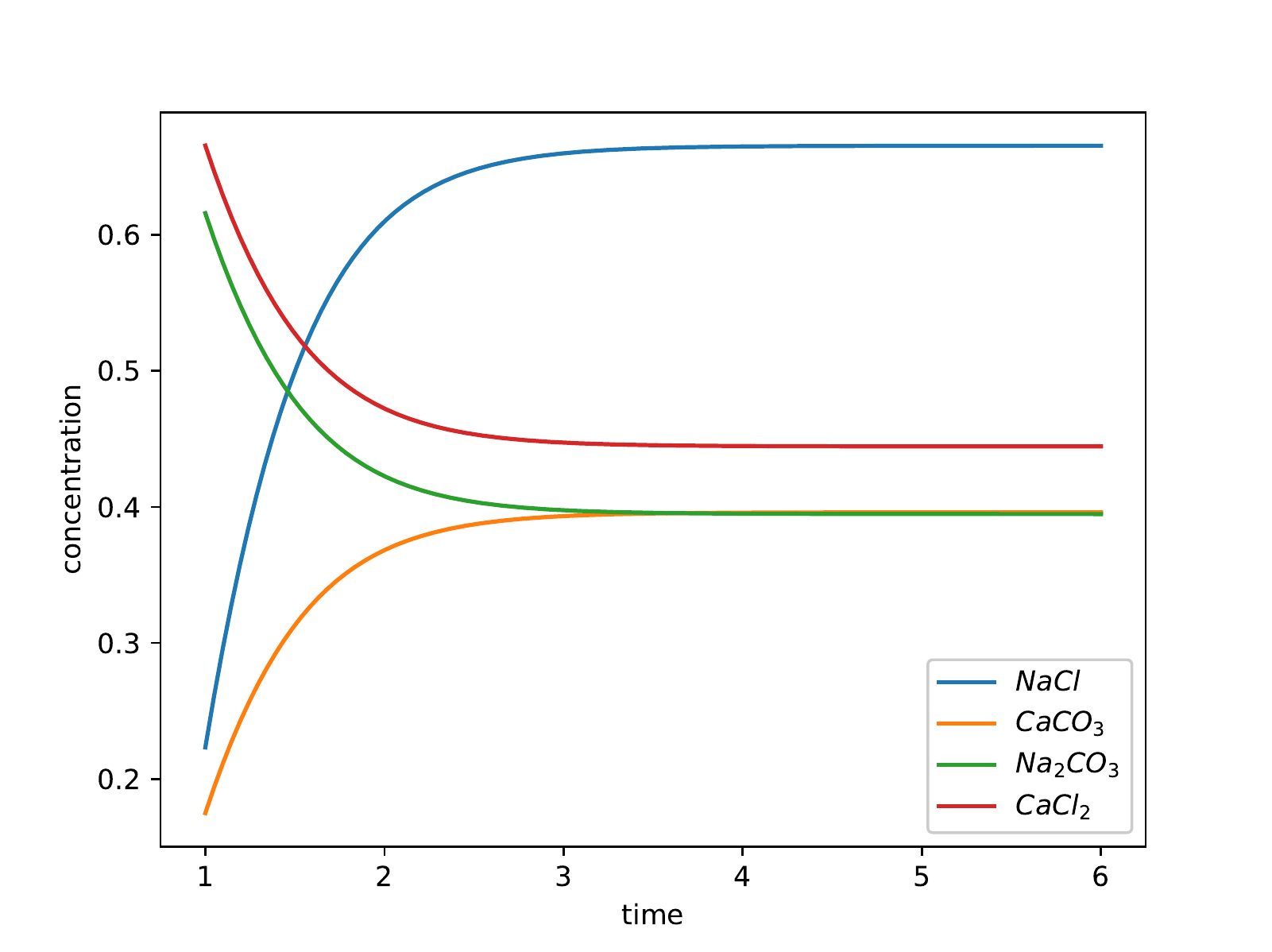}

}

\subfloat[Combustion engine example]{\includegraphics[scale=0.5]{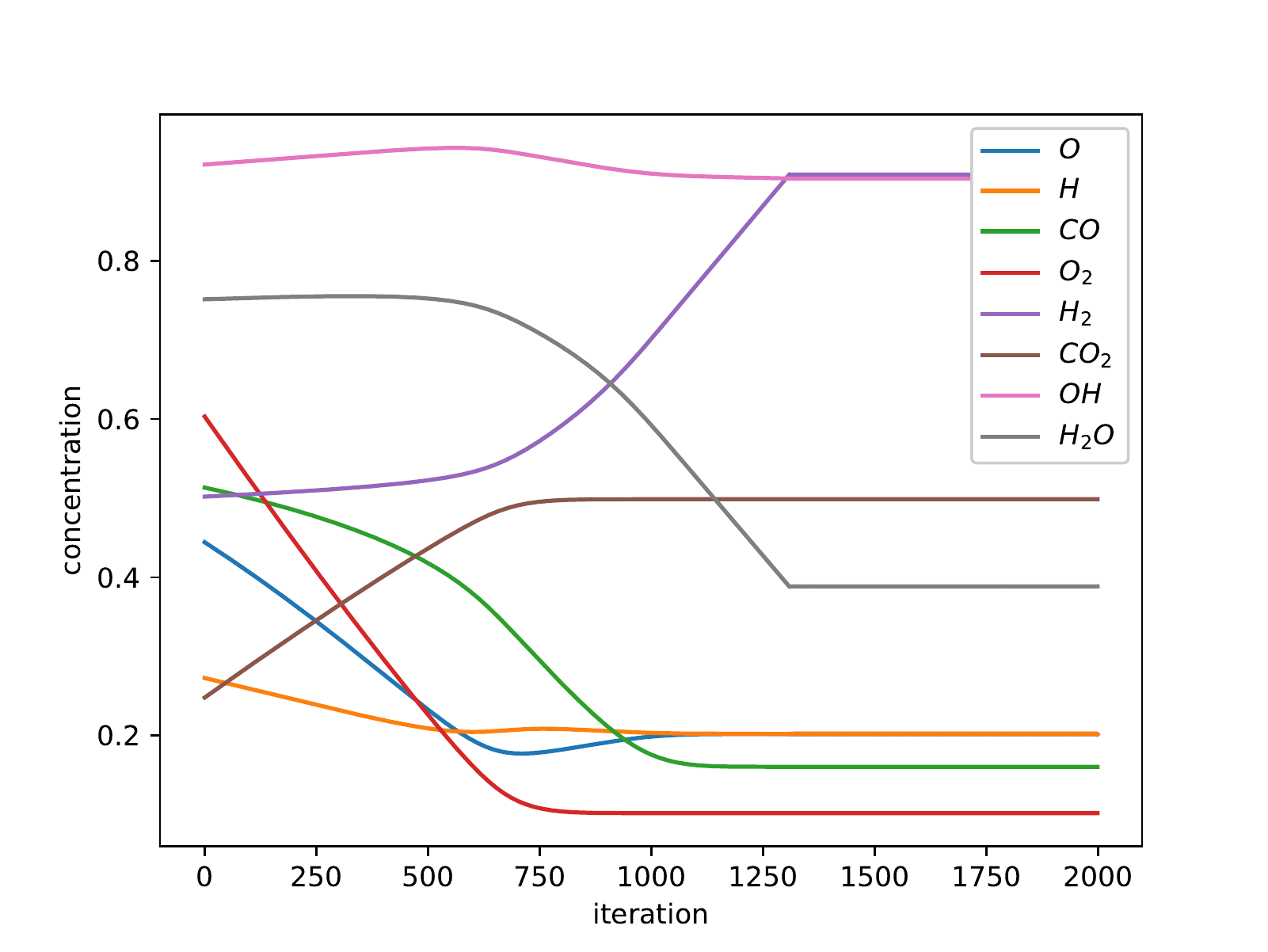}\includegraphics[scale=0.5]{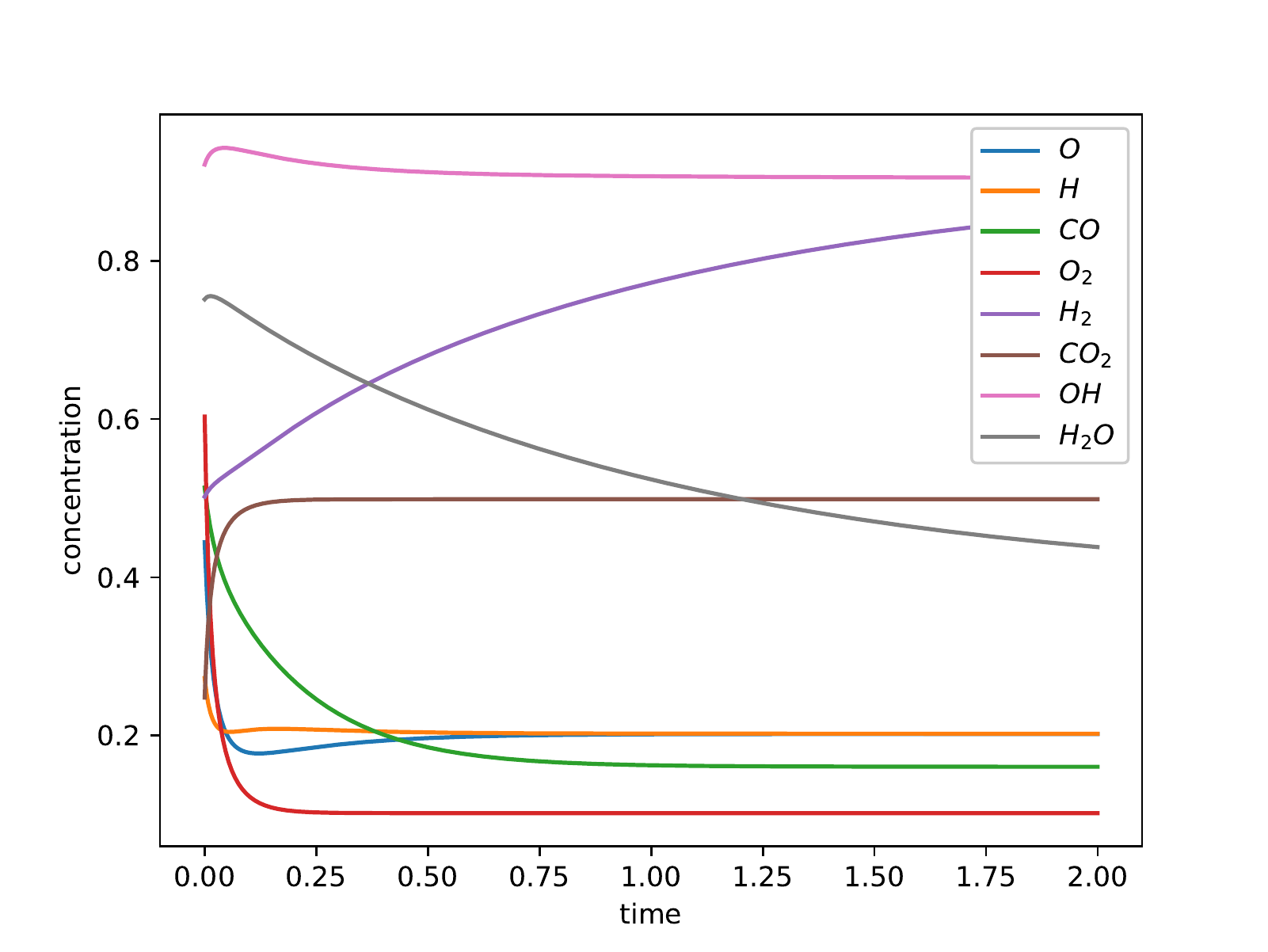}}\caption{Comparing convergence of projected gradient method and numerical solution
for different chemical reactions satisfying mass conservation law}
\label{chemical}
\end{figure}\vspace{0.1 cm}

\textbf{Example 2} ({\it A network of reversible reactions}) The example below is from combustion chemistry and models the burning of fuel in a combustion chamber such as in a car engine.
$$
O_2 \xrightleftharpoons[\, k_2\,]{\,k_1\,} 2O,~~H_2\xrightleftharpoons[\, k_4\,]{\,k_3\,} 2H,~~ N_2 \xrightleftharpoons[\, k_6\,]{\,k_5\,}2N,~~ CO_2 \xrightleftharpoons[\, k_8\,]{\,k_7\,} O+CO,~~ OH \xrightleftharpoons[\, k_{10}\,]{\,k_9\,}  O+H,~~H_2O \xrightleftharpoons[\, k_{12}\,]{\,k_{11}\,} O+2H.
$$
Let
\begin{align*}
&w_1=[O],~w_2=[H],~w_3=[CO],~w_4=[OH],~w_5=[O_2],\\
&~w_6=[H_2],~
~w_7=[H_2O],~w_8=[CO_2],~w_9=[N],~w_{10}=[N_2].
\end{align*}
It is remarked that this system can be split into one system of $(w_1,w_2,w_3,w_4,w_5,w_6,w_7,w_8)$ and one of $(w_9,w_{10})$, which are decoupled. For brevity, we will consider only the former system, since the latter is simpler (and can be treated similarly to {\bf Example 1}).
According to the law of mass action we obtain the following system of differential equations
\begin{equation}\label{sys1}
\begin{aligned}
    \dot{w}_1&=2(k_1 w_5-k_2w_1^2)+(k_7 w_8-k_8w_1w_3)+(k_9 w_4-k_{10}w_1w_2)+(k_{11}w_{7}-k_{12}w_1w_2^2)\\
    \dot{w}_2&=2(k_3w_6-k_4w_2^2)+(k_9w_4-k_{10}w_1w_2)+2(k_{11}w_{7}-k_{12}w_1w_2^2)\\
    \dot{w}_3&=(k_7w_8-k_8 w_1 w_3),\quad
    \dot{w}_4=-(k_9w_4-k_{7}w_1w_2)\\
    \dot{w}_5&=-(k_1 w_5-k_2w_1^2),\quad
    \dot{w}_6=-(k_3w_6-k_4w_2^2)\\
    \dot{w}_{7}&=-(k_{11}w_{7}-k_{12}w_1 w_2^2),\quad
    \dot{w}_8=-(k_7w_8-k_8w_1w_3).
\end{aligned}
\end{equation}
System \eqref{sys1} has three linearly independent conservation laws, for all $t\geq 0$,
\begin{equation}\label{laws_sys1}
\begin{aligned}
& w_1(t) + w_{4}(t) + 2w_5(t) + w_{7}(t) + w_8(t)  = M_1\\
& w_2(t) + w_4(t) + w_6(t)  + 2w_{7}(t) = M_2,\\
& w_3(t) + w_8(t) = M_3.
\end{aligned}
\end{equation}
By direct computations, for positive initial masses $M_1, M_2, M_3 >0$, there exists a unique positive equilibrium to \eqref{sys1} satisfying the conservation laws \eqref{laws_sys1} and solving
\begin{equation}\label{equi_sys1}
\begin{aligned}
	& k_2w_{1\infty}^2 = k_1w_{5\infty}, && k_{10}w_{1\infty}w_{2\infty} = k_9w_{4\infty},\\
	& k_4w_{2\infty}^2 = k_3w_{6\infty}, && k_{12}w_{1\infty}w_{2\infty}^2 = k_{11}w_{7\infty},\\
	& k_8w_{1\infty}w_{3\infty} = k_7w_{8\infty}. &&
\end{aligned}
\end{equation}
Therefore, system \eqref{sys1} is detailed balanced. Moreover, there exists no boundary equilibrium (direct verification), and it follows immediately from Proposition \ref{convergence} that any solution to \eqref{sys1} converges exponentially to the detailed balance equilibrium in the same stoichiometric compatibility class. The loss functions in this case are computed as
\begin{align*}
& \ell_1= -2(k_1 w_5 w_1-\frac{1}{3}k_2 w_1^3)-(k_7w_1w_8-\frac{1}{2}k_8 w_1^2 w_3)\\
&\quad \quad -(k_9w_1 w_4-\frac{1}{2}k_{10}w_1^2 w_2)-(k_{11}w_1w_{7}-\frac{1}{2}k_{12}w_1^2 w_2^2)
\\ & \ell_2=-2(k_3w_6w_2-\frac{1}{3}k_4 w_2^3)-(k_9 w_2w_4-\frac{1}{2}k_{10}w_1w_2^2)-2(k_{11}w_{7}w_2-\frac{1}{3}k_{12}w_1w_2^3)\\
&\ell_3=-(k_7w_3w_8-\frac{1}{2}k_8 w_1 w_3^2),\quad
\ell_4=(\frac{1}{2}k_1 w_5^2-k_2 w_1^2 w_5)\\
& \ell_5=(\frac{1}{2}k_3w_6^2-k_4w_2^2w_6),\quad
\ell_6=(\frac{1}{2}k_7w_8^2-k_8w_1w_3w_8)\\
&\ell_7=(\frac{1}{2}k_9 w_4^2-k_{10}w_1w_2w_4),\quad
\ell_{8}=(\frac{1}{2}k_{11}w_{7}^2-k_{12}w_1w_2^2 w_{7}).
\end{align*}
Under the detailed balance condition, the potential is given by
\begin{align*}
\mathcal{E}&=\sum_{i=1}^{8}w_i\Big(\log(w_i/\wb)-1\Big).
\end{align*}
\begin{align*}
&\Psi^\ast(w_1,\ldots,w_8;\mu_1,\ldots,\mu_8)\\&=\frac{\kappa_1}{2}\ell\Big(\frac{w_5}{w_{5\infty}},\frac{w_1^2}{w_{1\infty}^2}\Big)(\mu_5-2\mu_1)^2+\frac{1}{2}\kappa_2 \ell\Big(\frac{w_6}{w_{6\infty}},\frac{w_2^2}{w_{2\infty}^2}\Big)(\mu_6-2\mu_2)^2
\\&\qquad+\frac{1}{2}\kappa_4 \ell\Big(\frac{w_8}{w_{8\infty}},\frac{w_1w_3}{w_{1\infty}w_{3\infty}}\Big)(\mu_8-\mu_1-\mu_3)^2+ \frac{1}{2}\kappa_5 \ell\Big(\frac{w_4}{w_{4\infty}},\frac{w_1w_2}{w_{1\infty}w_{2\infty}}\Big)(\mu_4-\mu_1-\mu_2)^2\\
&\qquad +\frac{1}{2}\kappa_6 \ell\Big(\frac{w_{7}}{w_{7\infty}},\frac{w_1w_2^2}{w_{1\infty}w_{2\infty}^2}\Big)(\mu_{7}-\mu_1-2\mu_2)^2.
\end{align*}

Similar to the above example, using projected gradient method, the concentrations of all chemical
species in the corresponding chemical reaction converge to a close
neighborhood of equilibrium point ( Fig.\ref{chemical}(c,d)). In most cases, we observe the oscillations
around the equilibrium point. This is due to the finite numerical error when approaching the optimal points of loss functions. Using a smaller time step will reduce oscillations.

This example and the above example show the validity of the projected
gradient method in finding equilibria of the chemical reactions.

\medskip

{\bf Example 3} ({\it A reaction network without conservation laws}) 
\begin{equation*}
\begin{tikzpicture} [baseline=(current  bounding  box.center)]
\node (a) {$S_1+S_2$} node (b) at (2,0) {$3S_1$} node (c) at (2,-2) {$2S_1+S_3$} node (d) at (0,-2) {$2S_2$};

\draw[arrows=->] ([xshift =0.5mm]a.east) -- ([xshift =-0.5mm]b.west);
\draw[arrows=->] ([xshift =-0.5mm]b.west) -- ([xshift =0.5mm]a.east);

\draw[arrows=->] ([yshift=0.5mm]d.north) -- ([yshift=-0.5mm]a.south);
\draw[arrows=->] ([yshift=-0.5mm]a.south) -- ([yshift=0.5mm]d.north) ;

\draw[arrows=->] ([yshift=-0.5mm]b.south) -- ([yshift =0.5mm]c.north);
\draw[arrows=->] ([yshift =0.5mm]c.north) -- ([yshift=-0.5mm]b.south);

\draw[arrows=->] ([xshift =0.5mm]c.west) -- ([xshift =-0.5mm]d.east);
\draw[arrows=->] ([xshift =-0.5mm]d.east) -- ([xshift =0.5mm]c.west);
\end{tikzpicture}
\end{equation*}
For simplicity, we assume that all reaction rate constants are one. Denote by $w_1, w_2, w_3$ the concentrations of $S_1, S_2, S_3$ respectively. We obtain, thanks to the law of mass action
\begin{equation}\label{ODE-sys}
	\begin{aligned}
	\dot{w}_1  &= w_1w_2 - 3w_1^3 - w_1^2w_3 + 3w_2^2,\\
	\dot{w}_2  &= w_1^3 + 2w_1^2w_3 - 3w_2^2,\\ 
	\dot{w}_3  &= w_1^3 - 2w_1^2w_3 + w_2^2.
\end{aligned}
\end{equation}
This system does not have any conservation laws, and it therefore has for any initial data a unique positive detailed balance equilibrium $(w_{1\infty}, w_{2\infty}, w_{3\infty}) = (1,1,1)$. We remark that this system has infinitely many {\it boundary equilibria} of the form $(0,0,\alpha)$ for $\alpha > 0$. Thanks to the Global Attractor Conjecture in the case of one linkage class \cite{And11}, we have
\begin{equation*}
	\lim_{t\to\infty}(w_1(t), w_2(t), w_3(t)) = (w_{1\infty}, w_{2\infty}, w_{3\infty})
\end{equation*}
provided the initial data is positive. 
Moreover, one can even show that this convergence is exponential (see e.g. \cite[Proposition 2.3]{DFT17}). The loss functions in this case are computed as
\begin{equation}
\begin{aligned}\label{ls}
	& \ell_1 = \frac 12 w_1^2w_2 - \frac 34 w_1^4 - \frac 13 w_1^3 w_3 + 3w_1w_2^2,\\
	& \ell_2 = w_2w_1^3 + 2w_1^2w_2w_3 - w_2^3,\\
	& \ell_3 = w_1^3w_3 - w_1^2w_3^2 + w_2^2w_3.
\end{aligned}
\end{equation}


\begin{figure}
\subfloat[]{\includegraphics[scale=0.5]{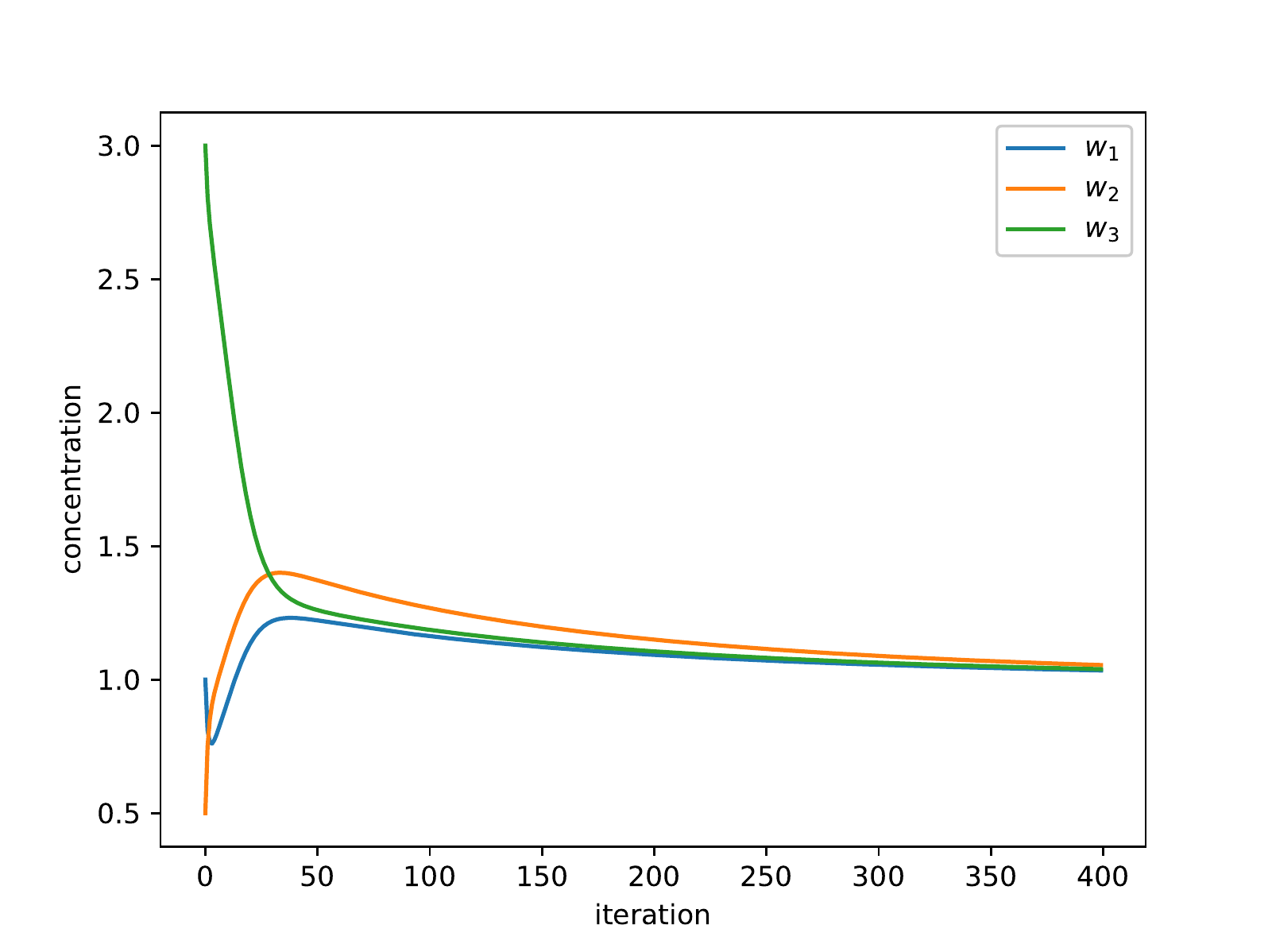}}
\subfloat[]{\includegraphics[scale=0.5]{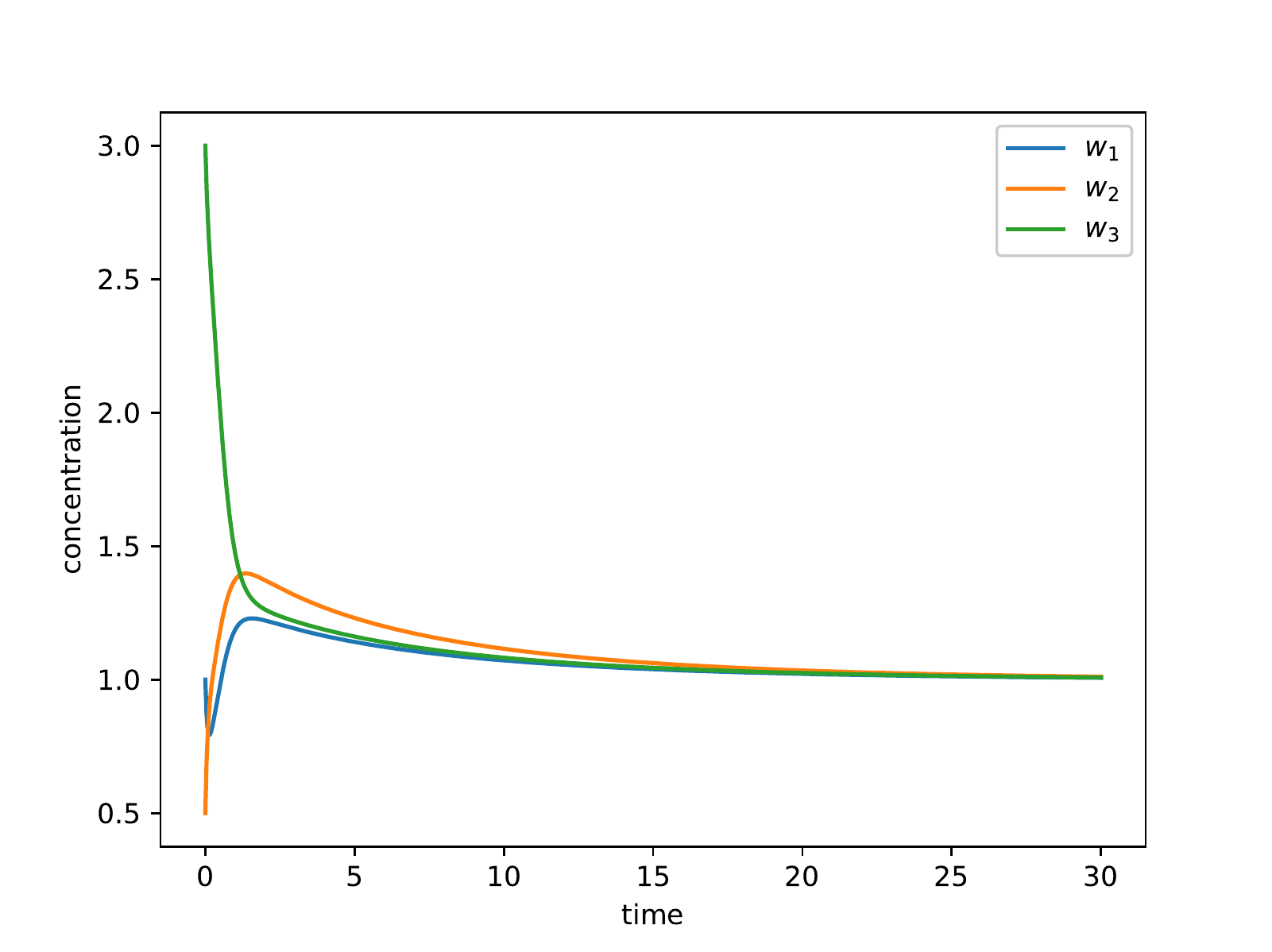}}\caption{Convergence of (a) gradient descend method and (b) numerical solution for chemical reaction that does not satisfy mass conservation law}
\label{chemical1}
\end{figure}
\noindent By applying gradient descend on the loss functions in (\ref{ls}), the concentrations of all chemical species in the corresponding chemical reaction converge to the equilibrium  point (Fig. \ref{chemical1}).
\section{Conclusion and Outlook}
\label{sec: summary}

In this paper, we develop the concept of generalized potential games enlarging significantly the class of potential games. We provide a necessary and sufficient condition for a game to be generalized potential using the connection of this concept to symmetrizable matrices. To show the applicability, we consider generalized potential games arising from chemical reaction network theory with detailed balance condition. Moreover, we use the projected steepest gradient descent method to calculate numerically the Nash equilibrium of generalized potential games and compare them with the evolution of the corresponding chemical reaction systems.

\medskip
The numerical comparison in this paper hints that there's a close connection between the trend to equilibrium of chemical reaction systems and the convergence of numerical methods for generalized potential games. More precisely, we believe that the convergence of a {\it discrete} version of this method (see e.g. \cite{Jungel2017}) leads to the convergence of algorithms finding equilibrium of the games defined in \eqref{eq: loss1}--\eqref{eq: lossn}. This interesting fact is contained in our forthcoming work.

\medskip
Another interesting direction of research is Hamiltonian games and decomposition of general games. Recently there has been considerable attempt in bringing concepts and methods from mechanics to game theory and machine learning \cite{balduzzi2018mechanics,Letcher2019,Bailey2019}. In particular, \cite{balduzzi2018mechanics,Letcher2019} have introduced a method to decompose a general multi-player game into a potential game and a Hamiltonian game based on the Helmholtz decomposition of the Hessian of the game. This resembles the decomposition of evolutionary dynamics into reversible and irreversible components in the celebrated GENERIC (General Equation for Non-Equilibrium Reversible and Irreversible Coupling) framework in non-equilibrium thermodynamics~\cite{Ottinger2005}. A great advantage of the GENERIC framework is that it not only automatically fulfills the laws of thermodynamics but also provides geometrical and physical structures. In this paper we have extended potential games to generalized potential games akin to the irreversible part of a GENERIC formulation. It would be interesting to further employ methodology and techniques from the GENERIC framework to understand and control
the dynamics in general games. We leave this topic for future research.

\appendix
\section{Explicit convergence to equilibrium}\label{appendix}
In this appendix, we show the explicit convergence to equilibrium for the single reversible reaction \eqref{reversible}. For simplicity, we assume without loss of generality that $k_f = k_r = 1$. Denote $a_i(t)$ and $b_j(t)$ as the concentrations of $A_i$ and $B_j$ respectively at time $t\geq 0$, for all $i=1, \ldots, m$ and $j=1,\ldots, n$. The general system \eqref{eq: CRE} in this case reads as
\begin{equation}\label{S}\tag{S}
\begin{cases}
a_{i}' = -\alpha_i(a^\alpha - b^\beta), &\text{ for } i=1,\ldots, m,\\
b_{j}' = \beta_j(a^\alpha - b^\beta),  &\text{ for } j=1,\ldots, n,\\
a_i(0) = a_{i0} > 0, &\text{ for } i=1,\ldots, m,\\
b_j(0) = b_{j0} > 0, &\text{ for } j=1,\ldots, n,
\end{cases}
\end{equation}
where we recall the convention \eqref{convention}. The solution to \eqref{S} obeys the following conservation laws
\begin{equation}\label{mass}
\frac{a_i(t)}{\alpha_i} + \frac{b_j(t)}{\beta_j} = M_{ij}:= \frac{a_{i0}}{\alpha_i} + \frac{b_{j0}}{\beta_j}, \quad \text{ for all } \quad i=1,\ldots, m, \; j=1,\ldots, n.
\end{equation}
The vector $\M = (M_{ij})_{i=1,\ldots, m,\, j=1,\ldots, n}$ is called the vector of initial masses. It is easy to see that there are precisely $m+n-1$ linearly independent conservations laws in \eqref{mass}. Thus, once we fix a suitable set of $m+n-1$ components of $\M$, then all other laws can be derived therefrom. For instance, $\M$ is completely defined if $M_{1j}$ and $M_{i1}$ are fixed for all $i=1,\ldots, m$ and $j=1,\ldots, n$. The following Lemma follows from Proposition \ref{pro:equilibrium}.

\begin{lemma}\label{lem:equilibrium}
	For any fix positive initial mass vector $\M$, there exists a unique positive chemical equilibrium $(a_\infty, b_\infty)\in (0,\infty)^{m+n}$ satisfying
	\begin{equation*}
	\begin{cases}
	a_\infty^\alpha = b_\infty^\beta,\\
	\dfrac{a_{i\infty}}{\alpha_i} + \dfrac{b_{j\infty}}{\beta_j} = M_{ij}, \text{ for all } i=1,\ldots, m, \, j=1,\ldots, n.
	\end{cases}
	\end{equation*}
\end{lemma}
To study the convergence to equilibrium for \eqref{S} we consider the entropy
\begin{equation}\label{entropy}
E[a,b] = \sum_{i=1}^{m}\varphi(a_i) + \sum_{j=1}^n\varphi(b_j)
\end{equation}
where $\varphi(z) = z\log z - z + 1$, the relative entropy
\begin{equation}\label{reentropy}
E[a,b|a_\infty,b_\infty] = E[a,b] - E[a_\infty,b_\infty] = \sum_{i=1}^m a_i\log\frac{a_i}{a_{i\infty}} - a_i + a_{i\infty} + \sum_{j=1}^nb_j\log\frac{b_j}{b_{j\infty}} - b_j + b_{j\infty},
\end{equation}
and its corresponding entropy dissipation
\begin{equation}\label{dissipation}
D[a,b] = -\frac{d}{dt}E[a,b|a_\infty,b_\infty] = (a^\alpha - b^\beta)\log\frac{a^\alpha}{b^\beta} \geq 0.
\end{equation}

\begin{proposition}\label{conv}
	Fix a positive initial mass vector $\M$ and let $(a_\infty,b_\infty)$ be the positive equilibrium defined by $\M$ in Lemma \ref{lem:equilibrium}. Then for any solution to \eqref{S} with initial mass $\M$ we have
	\begin{equation*}
	\lim_{t\to\infty}(a(t),b(t)) = (a_\infty,b_\infty).
	\end{equation*}
\end{proposition}
\begin{proof}
	Clearly $E[a,b|a_\infty,b_\infty]$ is a Lyapunov functional for \eqref{S} with the property that $E[a,b|a_\infty,b_\infty] = 0$ if and only if $(a,b) = (a_\infty,b_\infty)$ and
	\begin{equation*}
	\lim_{(a,b)\to \infty}E[a,b|a_\infty,b_\infty] = 0.
	\end{equation*}
	On the other hand, it's easy to check that $D[a,b] \geq 0$ for all $(a,b)$ and 
	\begin{equation*}
	D[a,b] = 0 \; \text{ and } \; (a,b) \text{ satisfies the mass conservation laws } \eqref{mass} \quad \Longleftrightarrow \quad (a,b) = (a_\infty,b_\infty).
	\end{equation*}
	Therefore, by Lyapunov stability theory we have the convergence
	\begin{equation*}
	\lim_{t\to\infty}(a(t),b(t)) = (a_\infty,b_\infty).
	\end{equation*}
\end{proof}

\begin{lemma}\label{lem:entropy}
	There exists an {\normalfont explicit} constant $C_1>0$ such that
	\begin{equation*}
	E[a,b|a_\infty,b_\infty] \geq C_1\left(\sum_{i=1}^m|a_i - a_{i\infty}|^2 + \sum_{j=1}^n|b_j - b_{j\infty}|^2 \right).
	\end{equation*}
\end{lemma}
\begin{proof}
	By using the elementary inequality $\varphi(z) = z\log z - z + 1 \geq (\sqrt{z} - 1)^2 = \frac{(z-1)^2}{(\sqrt{z}+1)^2}$ we have
	\begin{equation*}
	\begin{aligned}
	E[a,b|a_\infty,b_\infty] &= \sum_{i=1}^ma_{i\infty}\varphi\left(\frac{a_i}{a_{i\infty}}\right) + \sum_{j=1}^nb_{j\infty}\varphi\left(\frac{b_j}{b_{j\infty}} \right)\\
	&\geq \sum_{i=1}^ma_{i\infty}\frac{(a_i-a_{i\infty})^2}{(\sqrt{a_i} + \sqrt{a_{i\infty}})^2} + \sum_{j=1}^nb_{j\infty}\frac{(b_j - b_{j\infty})^2}{(\sqrt{b_j} + \sqrt{b_{j\infty}})^2}\\
	&\geq \frac{1}{M}\left(\sum_{i=1}^m|a_i-a_{i\infty}|^2 + \sum_{j=1}^n|b_j - b_{j\infty}|^2 \right)
	\end{aligned}
	\end{equation*}	
	where we used the estimate
	\begin{equation*}
	a_i \leq \alpha_i M_{i1} \quad \text{ and } \quad b_j \leq \beta_j M_{1j} \quad \text{ for all } i=1,\ldots, m,\; j=1,\ldots, n,
	\end{equation*}
	and the constant $M>0$ is defined as
	\begin{equation*}
	M = \max_{i=1,\ldots, m,\; j=1\ldots, n}\left\{\frac{1}{a_{i\infty}}(\sqrt{\alpha_iM_{i1}} + \sqrt{a_{i\infty}})^2; \; \frac{1}{b_{j\infty}}(\sqrt{\beta_jM_{1j}} + \sqrt{b_{j\infty}})^2\right\}.
	\end{equation*}
\end{proof}
\begin{proof}[Proof of Proposition \ref{pro:explicit}]
	We will employ the Bakry-Emery strategy. The proof is divided into several steps.
	
\noindent \textbf{Step 1.} Computing the derivative of the entropy dissipation \eqref{dissipation} we have
		\begin{equation*}
		\begin{aligned}
		\frac{d}{dt}D[a,b] &= \left[ \frac{d}{dt}a^\alpha - \frac{d}{dt}b^\beta\right]\log\frac{a^\alpha}{b^\beta} + (a^\alpha - b^\beta)\left[\sum_{i=1}^m\frac{a_i'}{a_i} - \sum_{j=1}^n\frac{b_j'}{b_j}\right]\\
		&= -\left[a^\alpha\sum_{i=1}^m\frac{\alpha_i^2}{a_i} + b^\beta\sum_{j=1}^n\frac{\beta_j^2}{b_j}\right](a^\alpha - b^\beta)\log\frac{a^\alpha}{b^\beta} - (a^\alpha - b^\beta)^2\left(\sum_{i=1}^m\frac{\alpha_i}{a_i} + \sum_{j=1}^n\frac{\beta_j}{b_j} \right)\\
		&\leq -\left[a^\alpha\sum_{i=1}^m\frac{\alpha_i^2}{a_i} + b^\beta\sum_{j=1}^n\frac{\beta_j^2}{b_j}\right]D[a,b].
		\end{aligned}
		\end{equation*}
		
\noindent \textbf{Step 2.} We prove that there exists an {\it explicit} constant $\lambda> 0$ such that
		\begin{equation}\label{lam}
		\Lambda(a,b):= a^\alpha\sum_{i=1}^m\frac{\alpha_i^2}{a_i} + b^\beta\sum_{j=1}^n\frac{\beta_j^2}{b_j} \geq \lambda.
		\end{equation}
		We will utilize the conservation laws and distinguish two cases:
		\begin{itemize}
			\item[(i)] if for some $i_0\in\{1,\ldots, m\}$, $a_{i_0} \leq \frac 12 \alpha_{i_0} \min_{j=1,\ldots, n}M_{{i_0}j}$, then for all $j=1,\ldots, n$,
			\begin{equation*}
			b_j = \beta_j\left(M_{{i_0}j} - \frac{a_{i_0}}{\alpha_{i_0}}\right)  \geq \frac 12\beta_{j}\min_{k=1,\ldots, n}M_{{i_0}k} =: K_j.
			\end{equation*}
			By now using $b_j \leq \beta_j \min_{i=1,\ldots, m}\{M_{ij} \}$ we have
			\begin{equation}\label{lam1}
			\Lambda(a,b) \geq b^\beta \sum_{j=1}^m\frac{\beta_j^2}{b_j} \geq \prod_{j=1}^mK_j^{\beta_j} \sum_{j=1}^m\frac{\beta_j}{\min_{i=1,\ldots, m}M_{ij}} > 0.
			\end{equation}
			\item[(ii)] if $a_i \geq \frac 12 \alpha_i \min_{j=1,\ldots, n}M_{ij} =: L_i$ for all $i=1,\ldots, m$ then we obtain
			\begin{equation}\label{lam2}
			\Lambda(a,b) \geq a^\alpha \sum_{i=1}^m\frac{\alpha_i^2}{a_i} \geq \prod_{i=1}^mL_i^{\alpha_i}\sum_{i=1}^{m}\frac{\alpha_i^2}{L_i} > 0.
			\end{equation}
			From \eqref{lam1} and \eqref{lam2} we get \eqref{lam} with
			\begin{equation}\label{lam_def}
			\lambda =  \min\left\{\prod_{j=1}^mK_j^{\beta_j} \sum_{j=1}^m\frac{\beta_j}{\min_{i=1,\ldots, m}M_{ij}}; \; \prod_{i=1}^mL_i^{\alpha_i}\sum_{i=1}^{m}\frac{\alpha_i^2}{L_i}\right\}.
			\end{equation}
		\end{itemize}
		
\noindent \textbf{Step 3.} It follows from the previous steps that
		\begin{equation}\label{diss_decay}
		\frac{d}{dt}D[a,b] \leq - \lambda D[a,b].
		\end{equation}
		Gronwall's inequality gives
		\begin{equation*}
		D[u,v](t) \leq e^{-\lambda t}D[a_0,b_0] \quad \text{ and  consequently } \quad \lim_{t\to\infty}D[u,v](t) = 0.
		\end{equation*}
		Using $D[a,b] = -\frac{d}{dt}E[a,b|a_\infty,b_\infty]$, integrating \eqref{diss_decay} on $(t,+\infty)$ and using Proposition \ref{conv}, leads to
		\begin{equation*}
		-D[a,b](t) \leq - \lambda E[a,b|a_\infty,b_\infty](t),s
		\end{equation*}
		or equivalently
		\begin{equation*}
		\frac{d}{dt}E[a,b|a_\infty,b_\infty] \leq -\lambda E[a,b|a_\infty,b_\infty].
		\end{equation*}
		Applying Gronwall's lemma and Lemma \ref{lem:entropy} yields
		\begin{equation*}
		\sum_{i=1}^m|a_i(t) - a_{i\infty}|^2 + \sum_{j=1}^n|b_j(t) - b_{j\infty}|^2 \leq C_1^{-1}E[a_0,b_0|a_\infty,b_\infty]e^{-\lambda t}		
		\end{equation*}
		for all $t>0$.
\end{proof}

\section{Projected steepest gradient descent algorithm}
Consider the constrained optimization of the form:
\begin{align*}
\min_{{\bf x}\in\mathbb{R}^{n}}&f({\bf x})\\
\text{s.t. } & A{\bf x}={\bf b}
\end{align*}
The projected steepest gradient descent algorithm to solve this problem is described in Algorithm \ref{alg:projected_gradient} \cite{freund2004projection}.
\begin{algorithm}
    \caption{Projected steepest gradient descent algorithm}
    \label{alg:projected_gradient}
    \begin{algorithmic}[1]
     \item Step 0: Initialize ${\bf x}$

	\item Step 1: Compute $\nabla f({\bf x})$.

	\item Step 2: Find moving direction
	\begin{itemize}
	\item Compute $P=I-A^{T}(AA^{T})^{-1}A$
	\item Compute $\beta=(\nabla f({\bf x})^{T})P\nabla f({\bf x})$
	\item If $\beta>0$, ${\bf d}=\frac{-P(\nabla f({\bf x}))}{\sqrt{\beta}}$ else if $\beta=0$, ${\bf d}=0$.
\item If $\nabla f({\bf x})^{T}{\bf d}=0$ stop.
\end{itemize}
\item Step 3: Update ${\bf x}={\bf x}+\alpha{\bf d}$. Go to step 1.
     \end{algorithmic}
\end{algorithm}

\medskip
{\bf Acknowledgement.} 
We thank Prof. Jeff Morgan for his careful proofreading of this manuscript.

This paper was completed during the visit of the third author at Hausdorff Research Institute for Mathematics in the Junior Trimester "Kinetic Theory". The institute's hospitality is greatly acknowledged. This work is partially supported by NAWI Graz and International Research Training Group IGDK 1754.

\bibliographystyle{alpha}
\bibliography{reference}
\end{document}